\documentclass[12pt]{article} 

\usepackage{amsfonts}
\usepackage{amsmath}
\usepackage{amscd}
\usepackage{amssymb}
\usepackage{amsthm}
\usepackage{amsxtra}
\usepackage{diagrams} 

\newcommand{\bbR}    {\mathbb R}

\newcommand{\bbZ}    {\mathbb Z}
\newcommand{\bbN}    {\mathbb N}
\newcommand{\bbI}    {\mathbb I}
\newcommand{\bbJ}    {\mathbb J}

\newcommand{\cA}    {{\cal A}}

\newcommand{\cC}    {{\cal C}}

\newcommand{\cF}    {{\cal F}}

\newcommand{\cH}    {{\cal H}}
\newcommand{\cI}    {{\cal I}}
\newcommand{\cJ}    {{\cal J}}

\newcommand{\cS}    {{\cal S}}
\newcommand{\cT}    {{\cal T}}

\newcommand{\al}    {\alpha}
\newcommand{\be}    {\beta}
\newcommand{\de}    {\delta}
\newcommand{\De}    {\Delta}
\newcommand{\eps}   {\epsilon}

\newcommand{\io}    {\iota}
\newcommand{\ka}    {\kappa}
\newcommand{\la}    {\lambda}
\newcommand{\ze}    {\zeta}
\newcommand{\om}    {\omega}

\newcommand{\Om}    {\Omega}

\newcommand{\te}    {\theta} 
\newcommand{\Te}    {\Theta}

\newcommand{\rE}     {\mathrm E}
\newcommand{\rF}     {\mathrm F}
\newcommand{\rf}     {\mathrm f}

\newcommand{\rg}     {\mathrm g}
\newcommand{\rH}     {\mathrm H}

\newcommand{\ri}     {\mathrm i}

\newcommand{\rj}     {\mathrm j}
\newcommand{\rK}     {\mathrm K}
\newcommand{\rk}     {\mathrm k}

\newcommand{\rl}     {\mathrm l}
\newcommand{\rM}     {\mathrm M}
\newcommand{\rN}     {\mathrm N}

\newcommand{\rT}     {\mathrm T}

\newcommand{\rV}     {\mathrm V}

\newcommand{\rX}     {\mathrm X}
\newcommand{\rx}     {\mathrm x}

\newcommand{\const}  {\mathrm{const}} 
\newcommand{\fin}    {\mathrm {fin}} 
\newcommand{\CE}     {\mathrm{CE}}
\newcommand{\PE}     {\mathrm{PE}}

\newcommand{\Sol}    {\mathrm {Sol}}

\newcommand{\bB}    {\mathbf B}

\newcommand{\mbf}   {\mathbf f}

\newcommand{\bp}    {\mathbf p}

\newcommand{\bx}    {\mathbf x}

\newcommand{\bu}    {\mathbf u}

\newcommand{\bj}    {\mathbf j}

\newcommand{\bCE}   {\mathbf{CE}}
\newcommand{\bPE}   {\mathbf{PE}}
\newcommand{\bCPE}  {\mathbf{CPE}}
\newcommand{\bNS}   {\mathbf{NS}}

\newcommand{\fD}    {\mathfrak D}

\newcommand{\w}     {\wedge}

\newcommand{\ck}    {\check}
\newcommand{\na}    {\nabla}

\newcommand{\bze}  {\pmb{\zeta}}

\newcommand{\p}     {\partial}

\DeclareMathOperator{\Hom}   {Hom}

\DeclareMathOperator{\End}   {End}

\DeclareMathOperator{\ord}   {ord}

\DeclareMathOperator{\im}    {Im}
\DeclareMathOperator{\ke}    {Ker}
\DeclareMathOperator{\com}   {{\scriptstyle\circ}\,}

\DeclareMathOperator{\graph} {graph}

\DeclareMathOperator{\ev}    {ev}

\DeclareMathOperator{\Sym}   {Sym}

\newtheorem{theorem}{Theorem}
\newtheorem{lemma}{Lemma}
\newtheorem{prop}{Proposition}

\theoremstyle{definition} 
\newtheorem{rem}{Remark}

\begin{document} 
\title{Navier--Stokes equations, the algebraic aspect} 
\author
          {
             Zharinov V.V.
             \thanks{Steklov Mathematical Institute, e-mail: zharinov@mi-ras.ru}
            }
\date{}
\maketitle 
\begin{abstract} 
Analysis of the Navier-Stokes equations in the frames of the algebraic approach 
to systems of partial differential equations (formal theory of differential equations) 
is presented. 
\end{abstract} 


\section{Preliminaries.} 

\subsection{Navier-Stokes equations}
The Navier-Stokes equations in physical notation are 
(see \cite{Sedov} -- \cite{FRR}, for example) 
\begin{align} 
	\bu_t+(\bu\cdot\na)\bu&=\nu\De\bu-\na p,   \label{NS1} 
	\\ 
	           \na\cdot\bu&=0, \label{NS2}
\end{align} 
where $\bu=(u^1,u^2,u^3)$ is the velocity field, $t$ is the time variable, the dot 
"$\cdot$" denotes the scalar product, $\na=(\na_1,\na_2,\na_3)$ is the gradient 
with respect to the spacial variables $\bx=(x^1,x^2,x^3)$, the parameter $\nu>0$ 
is the viscosity of the flow (do not confuse it with the index $\nu$), 
$\De=\na^2_1+\na^2_2+\na^2_3$ is the Laplacian, $p$ is the pressure. 

Here we study the Navier-Stokes equations from the algebra-geometrical point of view. 
The system (\ref{NS1})-(\ref{NS2}) is not formally integrable, as we need, to get the equivalent formally integrable system we add (see \cite{S}, for example) 
the trivial differential prolongations  
\begin{equation}\label{NS3}  
		\na\cdot\bu_t=0, \quad \na(\na\cdot\bu)=0, 
\end{equation} 
and the non-trivial differential prolongation ({\it hidden integrability condition})
\begin{equation}\label{NS4}  
	\De p+\na((\bu\cdot\na)\bu)=\De p+\na\bu\cdot\na\bu=0 	
\end{equation} 
(we have used the equation (\ref{NS2})). 

\begin{rem}\label{IC}
The equation (\ref{NS4}) is a Poisson equation for the pressure $p$ with the density 
$\rho=-\na\bu\cdot\na\bu$, it can be considered as the 
{\it inner constraint} for the Navier-Stokes equations.  
\end{rem}

\subsection{Notation}\label{NT} 

We use rather sophisticated notation: 
\begin{itemize} 
	\item 
		$\bbZ=\{0,\pm 1,\pm 2, \dots\}\supset\bbZ_+=\{0,1,2,\dots\}\supset\bbN=\{1,2,3,\dots\}$; 
	\item 
		$\rM=\overline{1,m}=\{1,2,\dots,m\}$, \quad $\rN=\overline{2,m}=\{2,\dots,m\}$, 
		\quad $\rM=\{1\}\cup\rN$; 
	\item		
		$\bbI=\bbZ^\rM_+=\{\ri=(i^1,\dots,i^m)\mid i^\mu\in\bbZ_+, \ \mu\in\rM\}$, 
		\quad $|\ri|=i^1+\dots+i^m$; 
	\item 
		$\ri+\rj=(i^1+j^1,\dots,i^m+j^m)$, \quad 
		$\ri+(\mu)=(i^1,\dots,i^\mu+1,\dots,i^m)$, \\ 
		$(\mu)=0+(\mu)$, for all $\ri,\rj\in\bbI$, $\mu\in\rM$; 
	\item 
		$\bbJ=\bbZ^\rN_+=\{\rj=(j^2,\dots,j^m)\mid j^\al\in\bbZ_+, \ \al\in\rN\}$; 
	\item 
		$\bbI=\bbZ_+\times\bbJ=\{\ri=(i^1,\rj)\mid i^1\in\bbZ_+, \ \rj\in\bbJ\}$; 
	\item 
		$\bbI_0=\{\ri\in\bbI\mid i^1=0\}=\{0\}\times\bbJ$, \quad  
		$\bbI_1=\{\ri\in\bbI\mid i^1=0,1\}=\{0,1\}\times\bbJ$;  
	\item 
		$\bbI'_0=\bbI\setminus\bbI_0=\{\ri\in\bbI\mid i^1>0\}$, \ 
		$\bbI'_1=\bbI\setminus\bbI_1=\{\ri\in\bbI\mid i^1>1\}$, \\ 
		$\bbI=\bbI_0\cup\bbI'_0=\bbI_1\cup\bbI'_1$.  
\end{itemize}
We shall also use the notation: 
\begin{itemize} 
	\item 
		$\rX=\bbR^\rM=\{\rx=(x^1,\dots,x^m)\mid x^\mu\in\bbR, \ \mu\in\rM\}$; 
	\item 
		$\bbR^\rM_\bbI=\{\bu=(u^\mu_\ri)\mid u^\mu_\ri\in\bbR, \ \mu\in\rM,
		\ \ri\in\bbI\}$, \\ 
		$\bbR^\rN_\bbI=\{\bu=(u^\al_\ri)\mid u^\al_\ri\in\bbR, \ \al\in\rN, 
		\ \ri\in\bbI\}$, \\  
		$\bbR^\rM_\bbI=\bbR^1_\bbI\times\bbR^\rN_\bbI
			=\bbR^1_{\bbI_0}\times\bbR^1_{\bbI'_0}\times\bbR^\rN_\bbI$;  
	\item 
		$\bbR_\bbI
			=\{\bp=(p_\ri)\mid p_\ri\in\bbR, \ \ri\in\bbI\}
			=\bbR_{\bbI_1}\times\bbR_{\bbI'_1}$; 
	\item 
		$\bB=\rX\times\bbR^\rM_\bbI\times\bbR_\bbI$, \\ 
		$\bCE=\rX\times\bbR^1_{\bbI_0}\times\bbR^\rN_\bbI\times
		\bbR_\bbI$, \\ 
		$\bCPE=\rX\times\bbR^1_{\bbI_0}\times\bbR^{\rN}_\bbI\times\bbR_{\bbI_1}$. 
\end{itemize} 
We assume the summation over repeated upper and lower indices in the prescribed limits. 

The paper is rather calculative, the calculations are straightforward  
but tiresome, and we omit them as a rule. 

\subsection{The base components} 

The base space for the Navier--Stokes equations in the algebraic approach 
is the infinite dimensional space $\bB=\rX\times\bbR^\rM_\bbI\times\bbR_\bbI$, 
where $\rX$ is the set of independent variables (space coordinates), 
$\bbR^\rM_\bbI$ is the set of differential variables (velocity coordinates 
and their partial derivatives), $\bbR_\bbI$ is the set of differential variables 
(pressure and its partial derivatives). 

The base algebra in the algebraic approach is the unital commutative associative algebra 
$\cA(\bB)=\cC^\infty_{\fin}(\bB)$ of all smooth real functions 
on the base space $\bB$ {\it of a finite order}, i.e., depending 
on a finite number of the variables $x^\mu,u^\mu_\ri,p_\ri$, where 
$\mu\in\rM$, $\ri\in\bbI$. 
In more detail, the integer $r\in\bbZ_+$ is called the {\it $\bu$-order} of a function 
$f(x,\bu,\bp)\in\cA(\bB)$, we write $\ord_\bu f=r$, if the partial derivative 
$\p_{u^\mu_\ri}f\ne0$ for some variable $u^\mu_\ri$, $|\ri|=r$, while partial derivatives 
$\p_{u^\mu_\ri}f=0$ for all $|\ri|>r$. 
In the same way, the $\bp$-order is defined.  

The base algebra of derivations of the algebra $\cA(\bB)$ is the Lie algebra 
$$
	\fD(\bB)=\fD(\cA(\bB))
	        =\big\{\ze=\ze^\mu\p_{x^\mu}+\ze^\mu_\ri\p_{u^\mu_\ri}+\ze_\ri\p_{p_\ri}
	        	\ \big| \ \ze^\mu,\ze^\mu_\ri,\ze_\ri\in\cA(\bB)\big\}. 
$$ 
The Lie algebra $\fD(\bB)$ splits into the {\it vertical} and {\it horizontal } subalgebras, 
$$ 
	\fD(\bB)=\fD_\rV(\bB)\oplus_{\cA(\bB)}\fD_\rH(\bB), 
$$ 
where 
\begin{itemize} 
	\item 
		$\fD_\rV(\bB)=\{\ze\in\fD(\bB)\mid \ze|_{\cC^\infty(\rX)}=0\}$ 
		\qquad\qquad\quad (note, $\cC^\infty(\rX)\subset\cA(\bB)$)\\
		\phantom{$\fD_\rV(\bB)$ }$=\{\ze=\ze^\mu_\ri\p_{u^\mu_\ri}+\ze_\ri\p_{p_\ri}
		            \mid \ze^\mu_\ri,\ze_\ri\in\cA(\bB)\}$;  
	\item 
		$\fD_\rH(\bB)=\{\ze=\ze^\mu D_\mu\mid \ze^\mu\in\cA(\bB)\}$, 
		\qquad\qquad\qquad\quad \ \ ($D_\mu|_{\cC^\infty(\rX)}=\p_{x^\mu}$)\\ 
		\phantom{$\fD_\rH(\bB)\!=\!\{\ze=\ze^\mu$}
		$D_\mu=\p_{x^\mu}+u^\la_{\ri+(\mu)}\p_{u^\la_\ri}+p_{\ri+(\mu)}\p_{p_\ri}$, 
		\qquad $\la,\mu\in\rM$;  
	\item 
		$[D_\mu,D_\nu]=0$, \quad 
		$[D_\mu,\ze]=(D_\mu\ze^\la_\ri-\ze^\la_{\ri+(\mu)})\p_{u^\la_\ri}
			        +(D_\mu\ze_\ri-\ze_{\ri+(\mu)})\p_{p_\ri}$, \\ 
		for all $\la,\mu,\nu\in\rM$, 
		$\ze=\ze^\la_\ri\p_{u^\la_\ri}+\ze_\ri\p_{p_\ri}\in\fD_\rV(\bB)$. 
\end{itemize}

We denote 
\begin{itemize} 
	\item 
		$D_\ri=(D_1)^{i^1}\com\dots\com(D_m)^{i^m}$ \quad for all \quad 
			$\ri=(i^1,\dots,i^m)\in\bbI$. 
\end{itemize}

The pair $(\cA(\bB),\fD_\rH(\bB))$ is called the 
{\it differential algebra} associated with the base space $\bB$. 

The Lie algebra 
\begin{equation*} 
	\Sym(\cA(\bB),\fD_\rH(\bB))
		=\big\{\ze=\ev_\rf\in\fD_\rV(\bB) \ \big| \ [D_\mu,\ev_\rf]=0, \ \mu\in\rM\big\} 
\end{equation*} 
is the {\it Lie algebra of symmetries} of the differential algebra 
$(\cA(\bB),\fD_\rH(\bB))$, 
where 
\begin{itemize} 
	\item 
		$\rf=(f^\mu,f)\in\cA^\rM(\bB)\times\cA(\bB)$, \quad  
		$f^\mu=\ze^\mu_0$, $f=\ze_0$, 
	\item 
		$\ev_\rf=D_\ri f^\mu\cdot\p_{u^\mu_\ri}+D_\ri f\cdot\p_{p_\ri}$, \quad 
		$D_\ri f^\mu=\ze^\mu_\ri$, $D_\ri f=\ze_\ri$. 
\end{itemize}  

The $\bbZ$-graded $\cA(\bB)$-module $\Om_\rH(\bB)=\oplus_{q\in\bbZ}\Om^q_\rH(\bB)$ 
of {\it horizontal differential forms} is defined as follows,  
\begin{equation*} 
	\Om^q_\rH(\bB)=
	        \begin{cases} 
	      0,                                          & q<0,q>m; \\
	      \cA(\bB),                                   & q=0;     \\
	      \Hom_{\cA(\bB)}(\w^q\fD_\rH(\bB);\cA(\bB)), & 1\le q\le m.  
	        \end{cases}
\end{equation*}
Here, 
\begin{align*} 
\Hom_{\cA(\bB)}&(\w^q\fD_\rH(\bB);\cA(\bB)) \\ 
             &=\big\{\om^q=\om_{\mu_1\dots\mu_q}\cdot dx^{\mu_1}\w\dots\w dx^{\mu_q}
	            \ \big| \ \om_{\mu_1\dots\mu_q}\in\cA(\bB), \text{+ s-s}\big\}, 
\end{align*} 
where the abbreviation "+s-s" states that components $\om_{\mu_1\dots\mu_q}$ 
are skew-symmetric in indices $\mu_1,\dots,\mu_q\in\rM$. 

The {\it horizontal differential} $d_\rH\in\End_\bbR(\Om_\rH(\bB))$, $d_\rH\com d_\rH=0$, 
is defined by the rule, 
\begin{align*} 
	d^q_\rH=d_\rH\big|_{\Om^q_\rH(\bB)} : \Om^q_\rH(\bB)&\to\Om^{q+1}_\rH(\bB), \\ 
	      \om_{\mu_1\dots\mu_q}\cdot dx^{\mu_1}\w\dots\w dx^{\mu_q}
	       &\mapsto 
	       D_{[\mu_0}\om_{\mu_1\dots\mu_q]}\cdot dx^{\mu_0}\w\dots\w dx^{\mu_q},  
\end{align*}
the brackets $[\dots]$ denote the skew-symmetrization in indices 
$\mu_0,\dots,\mu_q\in\rM$, so $d^{q+1}_\rH\com d^q_\rH=0$, $q\in\bbZ$. 

The $\bbZ$-graded linear space 
$H(\Om_\rH(\bB);d_\rH)=\oplus_{q\in\bbZ}H^q_\rH(\bB)$ of cohomologies 
of the differential algebra $(\cA(\bB);\fD_\rH(\bB))$ is defined in the usual way, 
$H(\Om_\rH(\bB);d_\rH)=\ke d_\rH\big/\im d_\rH$, $H^q_\rH(\bB)=\ke d^q_\rH\big/\im d^{q-1}_\rH$, 
$q\in\bbZ$. 

\begin{theorem}[{\bf The main theorem of the formal calculus of variations}, 
see, for example, \cite{O}, and the references therein]\label{T1} 
The linear spaces 
\begin{equation*} 
	H^q_\rH(\bB)=
	    \begin{cases} 
		0,        & q<0, q>m ;\\
		\bbR,     & q=0 ; \\
		0,        & 1\le q\le m-1 ; \\ 
		\cH(\bB), & q=m; 
	\end{cases}
\end{equation*} 
where (in our setting) the Helmholtz linear space 
$$ 
	\cH(\bB)=\big\{\chi=(\chi_\mu,\chi)\in\cA_\rM(\bB)\times\cA(\bB) \ \big| 
	\ \chi_*=\chi^*\big\}, 
$$  
the linear mappings 
	$\chi_*,\chi^* : \cA^\rM(\bB)\times\cA(\bB)\to\cA_\rM(\bB)\times\cA(\bB)$ 
act by the rules: 
\begin{align*} 
	&\cA^\rM(\bB)\times\cA(\bB)\ni\rf=(f^\mu,f)\mapsto
	\chi_*\rf=\rg=(g_\mu,g)\in\cA_\rM(\bB)\times\cA(\bB), \\ 
	&g_\mu=\p_{u^\nu_\ri}\chi_\mu\cdot D_\ri f^\nu+\p_{p_\ri}\chi_\mu\cdot D_\ri f, \quad 
	g=\p_{u^\nu_\ri}\chi\cdot D_\ri f^\nu+\p_{p_\ri}\chi\cdot D_\ri f; \\ 
	&\cA^\rM(\bB)\times\cA(\bB)\ni\rf=(f^\mu,f)\mapsto
	\chi^*\rf=\rg=(g_\mu,g)\in\cA_\rM(\bB)\times\cA(\bB), \\ 
	&g_\mu=(-D)_\ri\big(f^\nu\!\cdot\p_{u^\mu_\ri}\chi_\nu+f\!\cdot\p_{u^\mu_\ri}\chi\big), 
	\quad g=(-D)_\ri\big(f^\nu\!\cdot\p_{p_\ri}\chi_\nu
	+f\!\cdot\p_{p_\ri}\chi\big).   
\end{align*}
The isomorphism $\de=(\de_{u^\mu},\de_p) : H^m_\rH(\bB)\simeq\cH(\bB)$ of linear spaces 
is defined as the  variational derivatives: 
\begin{align*} 
	&\Om^m_\rH(\bB)\ni\om=L\cdot d^mx\mapsto\chi=\de L=(\de_{u^\mu}L,\de_p L), \\
	&\de_{u^\mu}L=(-D)_\ri\p_{u^\mu_\ri}L, \quad \de_p L=(-D)_\ri\p_{p_\ri}L. 
\end{align*}  
\end{theorem} 
\begin{rem} 
Note, cohomologies are defined up to isomorphisms. 
\end{rem}

\section{Constraints}

\subsection{The divergence-free space}\label{dfs}

The continuity equation $\CE=\{\p_{x^\mu}u^\mu=0\}$ (see (\ref{NS2})) 
has the algebraic counterpart 
$$ 
	\bCE=\big\{(x,\bu,\bp)\in\bB \ \big| \ \CE_\ri=u^\mu_{\ri+(\mu)}=0, \ \ri\in\bbI\big\}.
$$
The subset $\bCE\subset\bB$ is called the {\it solution manifold} of the equation 
$\CE$, because a smooth function $\phi(\rx)=(\phi^\mu(\rx))\in\cC^\infty(\rX;\bbR^\rM)$ 
is a solution of the equation $\CE$ iff the graph 
$$
	\graph\bj\phi(\rx)=\big\{(\rx,\bu,\bp) \ \big| \ u^\mu_\ri=\p_{x^\ri}\phi^\mu(\rx), 
	\ \mu\in\rM, \ \ri\in\bbI\big\}\subset\bCE, 
$$ 
where 
$\bj\phi(\rx)=\{\p_{x^\ri}\phi^\mu(\rx)\mid \ri\in\bbI, \ \mu\in\rM\}$ 
is the {\it jet} of the function $\phi(\rx)$, the variables $\bp=(p_\ri)\in\bbR_\bbI$ are parameters, 
$\p_{x^\ri}=(\p_{x^1})^{i^1}\com\dots\com(\p_{x^m})^{i^m}$, $\ri\in\bbI$.   

With the solution manifold $\bCE$ there is associated the ideal 
\begin{align*}
	\cI=\big\{f\in\cA(\bB) \ \big| \ f\big|_{\bCE}=0\big\}
	&=\big\{f=P(D)u^\mu_{(\mu)} \ \big| \ P(D)\in\cA(\bB)[D]\big\} \\ 
	&=\big\{f=\phi^\ri\cdot u^\mu_{\ri+(\mu)} \ \big| \ \phi^\ri\in\cA(\bB), \ \ri\in\bbI,
		+\fin\big\}  
\end{align*} 
of the algebra $\cA(\bB)$, where $\cA(\bB)[D]$ is the unital associative algebra 
of all polynomials in the indeterminates $D_1,\dots,D_m$ with coefficients in $\cA(\bB)$, 
i.e., 
$$ 
	\cA(\bB)[D]=\big\{P(D)=\phi^\ri\cdot D_\ri \ \big| \ \phi^\ri\in\cA(\bB), 
	 \ \ri\in\bbI, +\fin\big\},
$$ 
where the abbreviation "+fin" states that only a finite number of the coefficients 
$\phi^\ri\ne0$. The ideal $\cI$ is {\it differential}, i.e., 
$D_\mu\big|_{\cI} : \cI\to\cI$, $\mu\in\rM$. 

By definition, we assume $\cA(\bCE)=\cA(\bB)\big/\cI$, i.e., we define the algebra 
$\cA(\bCE)$ of {\it smooth functions on the space $\bCE$} as the factor-algebra. 
We denote 
$\bar f=f|_{\bCE}=f+\cI$ the restriction of the function $f\in\cA(\bB)$ 
to the solution manifold $\bCE$ and its equivalence class in $\cA(\bCE)$. 

We shall need: 
\begin{itemize} 
	\item  \label{DNS}
		the subalgebra of the Lie algebra $\fD(\bB)$   
\begin{align*}
		\fD_\cI(\bB)
		&=\big\{\ze\in\fD(\bB) \ \big| \ \ze|_{\cI} : \cI\to\cI \big\}
		 =\big\{\ze\in\fD(\bB) \ \big| \ \ze|_{\cI}\in\fD(\cI)\big\} \\
		&=\big\{\ze=\ze^\mu\p_{x^\mu}+\ze^\mu_\ri\p_{u^\mu_\ri}+\ze_\ri\p_{p_\ri}
	        	\ \big| \ \ze^\mu_{\ri+(\mu)}\in\cI, \ \ri\in\bbI \big\} 
\end{align*} 
		of all derivations of the algebra $\cA(\bB)$ leaving the ideal $\cI$ invariant;
	\item 
		the subalgebra of the Lie algebra $\fD(\bB)$ and the ideal 
		of the Lie algebra $\fD_\cI(\bB)\subset\fD(\bB)$
\begin{align*}
		\fD(\bB;\cI)
		&=\big\{\ze\in\fD(\bB) \ \big| \ \ze : \cA(\bB)\to\cI\big\} \\
		&=\big\{\ze=\ze^\mu\p_{x^\mu}+\ze^\mu_\ri\p_{u^\mu_\ri}+\ze_\ri\p_{p_\ri} 
	        	\ \big| \ \ze^\mu,\ze^\mu_\ri,\ze_\ri\in\cI\big\} 
\end{align*} 
		of all derivations with the image in the ideal $\cI$, 
	\item 
		the Lie factor-algebra $\fD(\bCE)=\fD_\cI(\bB)\big/\fD(\bB;\cI)$ 
		of derivations of the factor-algebra $\cA(\bCE)$. 
\end{itemize} 
Here, the mapping $\fD_\cI(\bB)\to\fD(\bCE)$ is defined as follows: 
\begin{align*} 
	\fD_\cI(\bB)\ni\ze\mapsto\bar\ze=\ze+\fD(\bB;\cI) : 
	   \cA(\bCE)&\to\cA(\bCE), \\ 
	\bar f=f+\cI&\mapsto\bar\ze\bar f=\overline{\ze f}=\ze f+\cI 
\end{align*} 
(we denote $\bar\ze=\ze+\fD(\bB;\cI)$ the equivalence class of a derivation 
$\ze\in\fD_\cI(\bB)$). 

The Lie algebra $\fD(\bCE)$ splits into the {\it vertical} and {\it horizontal} 
subalgebras: 
\begin{itemize} 
	\item 
		$\fD(\bCE)=\fD_\rV(\bCE)\oplus_{\cA(\bCE)}\fD_\rH(\bCE)$, 
	\item 
		$\fD_\rV(\bCE)
		=\{\bar\ze\in\fD(\bCE)\mid \ze\in\fD_\rV(\bB)\cap\fD_\cI(\bB)\}$, 
	\item 
		$\fD_\rH(\bCE)=\{\bar\ze\in\fD(\bCE)\mid \ze\in\fD_\rH(\bB)\}$ \quad\quad  
		(note, $\fD_\rH(\bB)\subset\fD_\cI(\bB)$). 
\end{itemize} 

The Lie algebra 
\begin{align*} 
	\Sym&(\cA(\bCE),\fD_\rH(\bCE))
		=\big\{\bar\ze\in\fD_\rV(\bCE) \ \big| \ [\bar D_\mu,\bar\ze]=\bar0=\cI, 
			\ \mu\in\rM\big\} \\
		&=\big\{\bar\ze \ \big| \ \ze=D_\ri f^\mu\cdot\p_{u^\mu_\ri}
			+D_\ri f\cdot\p_{p_\ri}, \ \ D_\mu f^\mu\in\cI  \big\} \\
		&=\big\{\bar\ze \ \big| \ \ze=\ev_\rf, \ \   
		\rf=(f^\mu,f)\in\cA^\rM(\bB)\times\cA(\bB),\  (D_\mu f^\mu)\big|_{\bCE}=0\big\}
\end{align*}
is the Lie algebra of symmetries of the differential algebra $(\cA(\bCE),\fD_\rH(\bCE))$.    

The $\bbZ$-graded $\cA(\bCE)$-module $\Om_\rH(\bCE)=\oplus_{q\in\bbZ}\Om^q_\rH(\bCE)$ 
of {\it horizontal differential forms} is defined as follows,  
\begin{equation*} 
	\Om^q_\rH(\bCE)=\begin{cases} 
	             0,                                        & q<0, q>m; \\
	             \cA(\bCE),                                 & q=0; \\
	             \Hom_{\cA(\bCE)}(\w^q\fD_\rH(\bCE);\cA(\bCE)), & 1\le q\le m.  
	             \end{cases}
\end{equation*}
Here, 
\begin{align*} 
\Hom_{\cA(\bCE)}&(\w^q\fD_\rH(\bCE);\cA(\bCE)) \\ 
                &=\big\{\bar\om^q=\bar\om_{\mu_1\dots\mu_q}\cdot 
                d\bar x^{\mu_1}\w\dots\w d\bar x^{\mu_q} 
                \ \big| \ \bar\om_{\mu_1\dots\mu_q}\in\cA(\bCE), 
                \text{+ s-s}\big\},                                                                                                                                                                                                                                                                                                                                                                                                                                                                                                                                                                                                                                                                                                                                                                                        
\end{align*} 
where $d\bar x^\mu(\bar D_\nu)=\bar\de^\mu_\nu=\de^\mu_\nu+\cI$, $\mu,\nu\in\rM$. 

The {\it horizontal differential} $d_\rH\in\End_\bbR(\Om_\rH(\bCE))$, $d_\rH\com d_\rH=0$, 
is defined by the rule, 
\begin{align*} 
	d^q_\rH=d_\rH\big|_{\Om^q_\rH(\bCE)} : \Om^q_\rH(\bCE)&\to\Om^{q+1}_\rH(\bCE), \\ 
	      \bar\om_{\mu_1\dots\mu_q}\cdot d\bar x^{\mu_1}\w\dots\w d\bar x^{\mu_q}
	       &\mapsto 
	       \bar D_{[\mu_0}\bar\om_{\mu_1\dots\mu_q]}\cdot 
	       d\bar x^{\mu_0}\w\dots\w d\bar x^{\mu_q},  
\end{align*}
where $d^{q+1}_\rH\com d^q_\rH=0$, $q\in\bbZ$. 

The $\bbZ$-graded linear space 
$H(\Om_\rH(\bCE);d_\rH)=\oplus_{q\in\bbZ}H^q_\rH(\bCE)$ of {\it cohomologies 
of the differential algebra} $(\cA(\bCE);\fD_\rH(\bCE))$ is defined 
in the usual way, 
$H(\Om_\rH(\bCE);d_\rH)=\ke d_\rH\big/\im d_\rH$, 
$H^q_\rH(\bCE)=\ke d^q_\rH\big/\im d^{q-1}_\rH$, $q\in\bbZ$. 

To move further, in particular to calculate the cohomology spaces, 
we introduce global coordinates on the solution manifold $\bCE$. 
Namely, the linear space $\bbR^\rM_\bbI$ we split as 
$$
	\bbR^\rM_\bbI=\bbR^1_\bbI\times\bbR^\rN_\bbI
	=\bbR^1_{\bbI_0}\times\bbR^1_{\bbI'_0}\times\bbR^\rN_\bbI
$$ 
(see Subsection \ref{NT}). Then, the global coordinates on $\bCE$ are defined 
by the isomorphism of the linear spaces 
\begin{align*}
	 \bCE=\big\{(x^\mu,u^\mu_\ri,p_\ri)\in\bB \ \big| \ u^\mu_{\ri+(\mu)}=0\big\}
	&\to \rX\times\bbR^1_{\bbI_0}\times\bbR^\rN_\bbI\times\bbR_\bbI, \\ 
		(x^\mu,u^1_\ri,u^{\al}_\ri,p_\ri)
	&\mapsto (x^\mu,u^1_{\ri_0},u^{\al}_\ri,p_\ri),   
\end{align*}
with the inverse mapping 
\begin{align*} 
	\rX\times\bbR^1_{\bbI_0}\times\bbR^\rN_\bbI\times\bbR_\bbI 
	&\to \big\{(x^\mu,u^\mu_\ri,p_\ri)\in\bB \ \big| \ u^\mu_{\ri+(\mu)}=0\big\}=\bCE, \\ 
	(x^\mu,u^1_{\ri_0},u^\al_\ri,p_\ri)
	&\mapsto(x^\mu,u^1_\ri,u^\al_\ri,p_\ri),   	
\end{align*}
where $u^1_{\ri'_0}=-u^\be_{\ri'_0-(1)+(\be)}$, $\mu=(1,\al)\in\rM$, $\ri_0\in\bbI_0$, 
$\al,\be\in\rN$, $\ri'_0\in\bbI'_0$.
In these global coordinates we have: 
\begin{itemize} 
	\item 
		$\bCE=\rX\times\bbR^1_{\bbI_0}\times\bbR^\rN_\bbI\times\bbR_\bbI$, 
		\qquad $\cA(\bCE)=\cC^\infty_{\fin}(\bCE)$; 
	\item 
		$\fD_\rV(\bCE)=\{\ze=\ze^1_{\ri_0}\p_{u^1_{\ri_0}}
		+\ze^\al_\ri\p_{u^\al_\ri}+\ze_\ri\p_{p_\ri}\mid 
		\ze^1_{\ri_0},\ze^\al_\ri,\ze_\ri\in\cA(\bCE)\}$; 
	\item 
		$D_\mu=\p_{x^\mu}+u^1_{\ri_0+(\mu)}\p_{u^1_{\ri_0}}
			+u^\al_{\ri+(\mu)}\p_{u^\al_\ri}+p_{\ri+(\mu)}\p_{p_\ri}$,  
			\quad $\mu\in\rM$, \\
		 where \ $u^1_{\ri_0+(1)}=-u^\al_{\ri_0+(\al)}$; 
	\item \label{SymCE}
		$\Sym(\cA(\bCE),\fD_\rH(\bCE))
		=\{\ze\!=\!\ev_\rf\in\fD_\rV(\bCE)\mid [D_\mu,\ev_\rf]\!=\!0, \ 
		\mu\in\rM\}$,  where 
		\begin{itemize} 
		\item 
			$\rf=(f^\mu,f)\in\cA^\rM(\bCE)\times\cA(\bCE)$, 
			\quad $D_\mu f^\mu=0$,
		\item
		$\ev_\rf=D_{\ri_0}f^1\cdot\p_{u^1_{\ri_0}}+D_\ri f^\al\cdot\p_{u^\al_\ri}
			+D_\ri f\cdot\p_{p_\ri}$; 
		\end{itemize}  
	\item 
		$\Om^q_\rH(\bCE)=\{\om=\om_{\mu_1\dots\mu_q}\cdot 
			dx^{\mu_1}\w\dots\w dx^{\mu_q}\mid 
			\om_{\mu_1\dots\mu_q}\in\cA(\bCE), \text{+ s-s}\}$, $1\le q\le m$. 
\end{itemize} 

In more detail, here the equation $D_\mu f^\mu=0$ takes the form 
\begin{equation}\label{ce2} 
	D_\mu f^\mu=\p_{x^\mu}f^\mu+u^1_{\ri_0+(\mu)}\p_{u^1_{\ri_0}}f^\mu  
		+u^\al_{\ri+(\mu)}\p_{u^\al_\ri}f^\mu 
		+p_{\ri+(\mu)}\p_{p_\ri}f^\mu =0. 
\end{equation} 
Let us suppose the function $f=(f^\mu)$ has the $\bp$-order $r=\max_{\mu\in\rM}\ord_\bp f^\mu$, then the term $\p_{p_{\ka,\ri}}f^\mu\cdot p_{\ka,\ri+(\mu)}$ 
will be of the $\bp$-order $r+1$, while all other terms of the equation (\ref{ce2}) will have 
$\bp$-orders $\le r$. This leads to the conclusion $\p_{p_{\ka,\ri}}f^\mu=0$, $|\ri|=r$, hence
the $\bp$-order of the function $f$ is $\le r-1$. 
By induction we get $f^\mu=f^\mu(x,\bu)$, $\mu\in\rM$, i.e., $f$ does not depend on 
the variable $\bp\in\bbR_\bbI$, and we get the following statement. 
\begin{prop} \label{sce} 
In the above setting the equation (\ref{ce2}) is reduced to 
\begin{equation} 
	\p_{x^\mu}f^\mu+u^1_{\ri_0+(\mu)}\p_{u^1_{\ri_0}}f^\mu  
		+u^\al_{\ri+(\mu)}\p_{u^\al_\ri}f^\mu=0, 
		\quad f^\mu(x,\bu)\in\cC^\infty_{\fin}(\rX\times\bbR^1_{\bbI_0}\times\bbR^\rN_\bbI). 
\end{equation} 
\end{prop} 
\begin{rem}\label{rce} 
If we want to have symmetries that really depend on the variable 
$\bp\in\bbR_\bbI$, we must narrow down to a subvariety of the space $\bCE$. 
Exactly such situation arises when studying the Navier-Stokes equations. 
\end{rem} 

To calculate the cohomology spaces $H^q_\rH(\bCE)$ we introduce 
the auxiliary complex $\{\Te^q,d^q_\Te\mid q\in\bbZ\}$, 
where \label{AC1}
\begin{itemize} 
	\item 
		$\Te^q=\{\te^q=\te_{\al_1\dots\al_q}\!\!\cdot 
		dx^{\al_1}\w\dots\w dx^{\al_q}\mid \al_1,\dots,\al_q\in\rN, 
		\te_{\al_1\dots\al_q}\in\cA(\bCE)\}$; 
	\item 
		$d^q_\Te : \Te^q\to\Te^{q+1}, \  
		\te_{\al_1\dots\al_q}\!\!\cdot dx^{\al_1}\w\dots\w dx^{\al_q}
		\mapsto 
		D_{[\al_0}\te_{\al_1\dots\al_q]}\cdot 
		dx^{\al_0}\w\dots\w dx^{\al_q}$; 
	\item 
		$H^q(\Te)=\ke d^q_\Te\big/\im d^{q-1}_\Te$.  
\end{itemize} 
There is defined the commutative diagram  
\begin{equation*} 
\begin{diagram}[2em] 
 &    &0                 &\rTo^{\io^{-1}} &\Om^0           &\rTo^{\pi^0}&\Te^0             &\rTo&0 \\ 
 &    &\dTo^{d^{-1}_\Te} &                &\dTo^{d^0_\rH}  &            &\dTo^{d^0_\Te}    &    &  \\ 
 &    &\vdots            &                &\vdots          &            &\vdots            &    &  \\ 
 &    &\dTo^{d^{q-2}_\Te}&                &\dTo^{d^{q-1}\rH}&           &\dTo^{d^{q-1}_\Te}&    &  \\ 
0&\rTo&\Te^{q-1}         &\rTo^{\io^{q-1}}&\Om^q           &\rTo^{\pi^q}&\Te^q             &\rTo&0 \\ 
 &    &\dTo^{d^{q-1}_\Te}&                &\dTo^{d^{q}_\rH}&            &\dTo^{d^{q}_\Te}  &    &  \\
 &    &\vdots            &                &\vdots          &            &\vdots            &    &  \\     
 &    &\dTo^{d^{m-2}_\Te}&                &\dTo^{d^{m-1}_\rH}&          &\dTo^{d^{m-1}_\Te}&    &  \\  
0&\rTo&\Te^{m-1}         &\rTo^{\io^{m-1}}&\Om^m           &\rTo^{\pi^m}&0                 &    &  \\ 
\end{diagram}	
\end{equation*}
where 
\begin{itemize} 
	\item 
		$\Om^q=\Om^q_\rH(\bCE)$, \quad $\Om^0=\Te^0=\cA(\bCE)$; 
	\item 
		$\Om^q=dx^1\w\Te^{q-1}\oplus_{\cA(\bCE)}\Te^q$; 
	\item 
		$d_\rH(dx^1\w\te^{q-1}+\te^q)=dx^1\w(D_1\te^q-d_\Te\te^{q-1})+d_\Te\te^q$; 
	\item 
		$D_1\te^q=D_1(\te_{\al_1\dots\al_q}\cdot dx^{\al_1}\w\dots\w dx^{\al_q})
			=(D_1\te_{\al_1\dots\al_q})\cdot dx^{\al_1}\w\dots\w dx^{\al_q}$; 
	\item 
		$\io^q : \Te^q\to\Om^{q+1}$, \quad $\te^q\mapsto\om^{q+1}=(-1)^q dx^1\w\te^q+0$; 
	\item 
		$\pi^q : \Om^q\to\Te^q$, \quad $\om^q=dx^1\w\te^{q-1}+\te^q\mapsto\te^q$; 
\end{itemize} 

According to the general results of homology theory (see, for example, \cite{M}) 
this diagram defines the long exact sequence of the cohomology spaces 
\begin{align*} 
	&0\to H^0_\rH(\bCE)\to H^0(\Te)\xrightarrow{D^0_1}H^0(\Te)\to H^1_\rH(\bCE)
	\to H^1(\Te)\xrightarrow{D^1_1}\dots \\
	     &\dots\xrightarrow{D^{m-2}_1}H^{m-2}(\Te)\to H^{m-1}_\rH(\bCE)\to H^{m-1}(\Te) 
	\xrightarrow{D^{m-1}_1}H^{m-1}(\Te)\to \\
	     &\to H^m_\rH(\bCE)\to 0,  
\end{align*}
where $D^q_1 : H^q(\Te)\to H^q(\Te)$, 
$$
	[\te^q]=\te^q+\im d^{q-1}_\Te\mapsto D^q_1[\te^q]=[D_1\te^q]=D_1\te^q+\im d^{q-1}_\Te. 
$$
By Theorem \ref{T1}, we have $H^q(\Te)=0$ for $q\ne0$ and $q\ne m-1$, while 
\begin{itemize} 
	\item 
		$H^0(\Te)=\{\phi(x^1)\in\cC^\infty(\bbR)\}$; 
	\item   
		$H^{m-1}(\Te)=\cH(\Te)=\{\chi\in\cF(\Te)\mid \chi_*=\chi^*\}$,  \\  
		\phantom{$H^{m-1}(\Te)=$}$\cF(\Te)=\cA_1(\bCE)\times\!\cA^{\bbZ^1_+}_{\rN}(\bCE)
		\times\!\cA^{\bbZ^1_+}(\bCE)$,  
\end{itemize} 
where the isomorphism $\de : H^{m-1}(\Te)\simeq\cH(\Te)$ is defined as follows: 
\begin{align*} 
	&\Te^{m-1}\ni\te=L\cdot dx^2\w\dots\w dx^m\mapsto
		\de L=(\de_{u^1_0}L,\de_{u^\al_{i^1}}L,\de_{p_{i^1}} L), \\
	&\de_{u^1_0}L=(-D)_\rj\p_{u^1_{0,\rj}}L, 
	\quad \de_{u^\al_{i^1}}L=(-D)_\rj\p_{u^\al_{i^1,\rj}}L, 
	\quad \de_{p_{i^1}} L=(-D)_\rj\p_{p_{i^1,\rj}}L,   
\end{align*} 
here $\rj\in\bbJ$, $\al\in\rN$, $i^1\in\bbZ_+$, remind, $\ri_0=(0,\rj)$, $\ri=(i^1,\rj)$ 
(see Subsection \ref{NT}), so $u^1_{\ri_0}=u^1_{0,\rj}$, $u^\al_\ri=u^\al_{i^1,\rj}$, 
$p_\ri=p_{i^1,\rj}$. 

The above long exact sequence defines the short exact sequence 
$$
	0\to H^0_\rH(\bCE)\to\ke D^0_1\to 0, \quad\text{hence,}\quad H^0_\rH(\bCE)=\bbR. 
$$ 
In the same way, we have the exact sequence 
$$ 
	0\to\im D^0_1\to H^0(\Te)\to H^1_\rH(\bCE)\to 0 \quad (H^1(\Te)=0),
$$ 
hence, $H^1_\rH(\bCE)=H^0(\Te)\big/\im D^0_1=0$. 
Further, for $1\le q\le m-2$ we have $H^q(\Te)=0$ and 
$H^{q-1}(\Te)\to H^q_\rH(\bCE)\to H^q(\Te)$, so $H^q_\rH(\bCE)=0$ for $2\le q\le m-2$. 
Now, for $q=m-1$ we have 
$$	
	\dots H^{m-2}(\Te)\to H^{m-1}_\rH(\bCE)\to H^{m-1}(\Te)\xrightarrow{D^{m-1}_1}\dots 
$$
Here, $H^{m-2}(\Te)=0$, hence, 
$0\to H^{m-1}_\rH(\bCE)\to\ke D^{m-1}_1\to 0$ is the exact sequence, i.e., $H^{m-1}_\rH(\bCE)=\ke D^{m-1}_1$ 
(note, cohomologies are defined up to isomorphisms). 

\begin{lemma}\label{L1}{\rm(See, for example, \cite{Z1}, \cite{Z2})} 
There is defined the commutative diagram 
\begin{equation*}
\begin{diagram}[2em] 
	0&\rTo&H^{m-1}(\Te)    &\rTo^\de&\cH(\Te)          &\rTo&0 \\
     &	  &\dTo^{D^{m-1}_1}&        &\dTo^{\tilde{D_1}}&    &  \\ 
	0&\rTo&H^{m-1}(\Te)    &\rTo^\de&\cH(\Te)          &\rTo&0,      
\end{diagram} 
\end{equation*}  
in particular, $\de\com D^{m-1}_1=\tilde D_1\com\de$, 
where $D_1=\p_{x^1}+\ev'_{\rf}, \ \tilde{D_1}=D_1+\rf^*$, 
$$ 
	\ev'_{\rf}=D_\rj f^1_0\cdot\p_{u^1_{0,\rj}}+D_\rj f^\al_{i^1}\cdot\p_{u^\al_{i^1,\rj}} 
	+D_\rj f_{i^1}\cdot\p_{p_{i^1,\rj}}, 
$$  
$\rj\in\bbJ$, \ $\al\in\rN$, \ $i^1\in\bbZ_+$. 
\end{lemma}
In our setting, 
$\rf=(f^1_0=-u^\al_{(\al)}, \ f^\al_{i^1}=u^\al_{i^1+1,0}, \ 
			f_{i^1}=p_{i^1+1,0})$. 
Hence, here 
$\tilde{D_1},\rf^*\in\End_\bbR(\cA_1(\bCE)\times\cA^{\bbZ^1_+}_{\rN}(\bCE)\times
\cA^{\bbZ^1_+}(\bCE))$, 
$$ 
	\chi=(\chi^0_1,\chi^{i^1}_{\al},\chi^{i^1})\mapsto \rf^*\chi
	=((\rf^*\chi)^0_1,(\rf^*\chi)^{i^1}_{\al},(\rf^*\chi)^{i^1}, 
$$
where 
$$ 
	(\rf^*\chi)^0_1=0, \quad 
	(\rf^*\chi)^{i^1}_\al=\de^{i^1}_0 D_\al\chi^0_1+\chi^{i^1-1}_\al, \quad 
	(\rf^*\chi)^{i^1}=\chi^{i^1-1}.
$$
The system $(D_1+\rf^*)\chi=0$ of the equations for defining an unknown function   
$\chi\in\cA_1(\bCE)\times\cA^{\bbZ^1_+}_{\rN}(\bCE)\times\cA^{\bbZ^1_+}(\bCE)$ 
with a finite number of the components $\chi^0_1,\chi^{i^1}_\al,\chi^{i^1}\ne0$, 
reduces to  
$$ 
	D_1\chi^0_1=0, \quad D_1\chi^{i^1}_\al+\de^{i^1}_0 D_\al\chi^0_1+\chi^{i^1-1}_\al=0, 
	\quad D_1\chi^{ i^1}+\chi^{i^1-1}=0, 
$$ 
and has the general solution: $\chi=(\chi^0_1,0,0)$, $\chi^0_1\in\bbR$. 
Now, we may take $\de u^1=(1,0,0)$, hence the linear space $\ke D^{m-1}_1$ is one-dimensional with 
the basis $[u^1d_1x]$ (remind, $d_\mu x=(-1)^{\mu-1}dx^1\w\dots\ck{dx^\mu}\dots\w dx^m$, 
the term $\ck{dx^\mu}$ is omitted, so $dx^\nu\w d_\mu x=\de^\nu_\mu\cdot d^m x$, $\mu,\nu\in\rM$), i.e., $\ke D^{m-1}_1=\bbR\cdot[u^1d_1x]$ (cf., \cite{Z1}, \cite{Z2}).  
The condition $\chi_*=\chi^*$ is trivially fulfilled, hence the following statement is valid. 
\begin{prop} 
The cohomology space $H^{m-1}_\rH(\bCE)=\bbR\cdot[u^\mu d_\mu x]$. 
\end{prop}			
Note, here $H^{m-1}_\rH(\bCE)\ni[u^\mu d_\mu x]\mapsto[u^1d_1x]\in\ke D^{m-1}_1$. 

\begin{rem} 
The cohomology $[u^\mu d_\mu x]$ is generated by the constraint $\CE=D_\mu u^\mu=0$.   
\end{rem}

At last, we have $H^{m-1}(\Te)\xrightarrow{D^{m-1}_1}H^{m-1}(\Te)\to H^m_\rH(\bCE)\to0$, i.e., 
$$ 
	0\to\im D^{m-1}_1\to H^{m-1}(\Te)\to H^m(\bCE)\to 0, 
$$
so $H^m_\rH(\bCE)=H^{m-1}(\Te)\big/\im D^{m-1}_1$
Thus, the following theorem holds. 
\begin{theorem}\label{TCE} 
The linear spaces of the cohomologies of the differential algebra $(\cA(\bCE);\fD_\rH(\bCE))$ are 
\begin{equation*} 
	H^q_\rH(\bCE)=\begin{cases} 
	0,                                     & q<0,1\le q\le m-2,q>m; \\ 
	\bbR,                                  & q=0; \\ 
	\ke D^{m-1}_1=\bbR\cdot[u^\mu d_\mu x],& q=m-1; \\ 
	H^{m-1}(\Te)\big/\im D^{m-1}_1,        & q=m; 
	            \end{cases} 
\end{equation*} 
where $H^{m-1}(\Te)\big/\im D^{m-1}_1=\cH(\Te)\big/\im\tilde{D_1}$. 
\end{theorem} 

\subsection{An additional constraint.}\label{ADC} 

According to Remark \ref{rce}, to apply the algebraic analysis to the Navier-Stokes equations,  
explicitly containing dependence on the pressure $p$, we need to add an additional constraint 
in the space $\bCE$. Having in mind the specific form of these equations, 
we choose the constraint defined by the equation 
$$
	\PE=\De p+\na\bu\cdot\na\bu
	=\de^{\la\mu}\p_{x^\la}\p_{x^\mu}p+\p_{x^\la}u^\mu\cdot\p_{x^\mu}u^\la=0
$$ 
(see the integrability condition (\ref{NS4}), we took into account the equation
(\ref{NS2})). 
Namely, we set   
\begin{equation} 
	\bPE=\big\{(x,\bu,\bp)\in\bB \ \big| \ \PE_\ri=
	\De p_\ri+D_\ri(u^\la_{(\mu)}u^\mu_{(\la)})=0, \  
	\ri\in\bbI\big\}.   
\end{equation} 
where 
\begin{itemize} 
	\item   
		$\De=\de^{\la\mu}D_\la\com D_\mu
		=\de^{\la\mu}D_{(\la)+(\mu)}=\sum_\mu D_{2(\mu)}$, 
	\item 
		$D_\ri(u^\la_{(\mu)}u^\mu_{(\la)})
		=\sum_{\rk+\rl=\ri}\binom{\ri}{\rk}u^\la_{\rk+(\mu)}
		u^\mu_{\rl+(\la)}$.
\end{itemize}
The subspace $\bCPE=\bCE\cap\bPE\subset\bB$ generates the differential ideal 
$$ 
	\cJ=\big\{f\in\cA(\bB) \ \big| \ f\big|_{\bCPE}=0\big\}
	=\big\{f=g^\ri\cdot\CE_\ri+h^\ri\cdot\PE_\ri \ \big| \ g^\ri,h^\ri\in\cA(\bB)\big\}. 
$$ 
The factor-algebra $\cA(\bCPE)=\cA(\bB)\big/\cJ$ is the algebra of {\it smooth functions 
on the space} $\bCPE$. 

The Lie algebra $\fD(\bCPE)$ of {\it derivations of the algebra $\cA(\bCPE)$}, 
the Lie algebra $\Sym(\cA(\bCPE),\fD_\rH(\bCPE))$ of the {\it symmetries of the 
differential algebra} $(\cA(\bCPE),\fD_\rH(\bCPE))$, the $\bbZ$-graded $\cA(\bCPE)$-module 
$\Om_\rH(\bCPE)=\oplus_{q\in\bbZ}\Om^q_\rH(\bCPE)$ of {\it horizontal differential forms} 
and the {\it horizontal differential} $d_\rH\in\End_\bbR(\Om_\rH(\bCPE))$ on the space $\bCPE$ 
are defined as above (see Subsection \ref{dfs}). Namely, here: 
\begin{itemize} 
	\item 
		$\fD_\cJ(\bB)=\{\ze\in\fD(\bB)\mid \ze|_\cJ : \cJ\to\cJ\}$; 
	\item 
		$\fD(\bB;\cJ)=\{\ze\in\fD(\bB)\mid \ze : \cA(\bB)\to\cJ\}$; 
	\item 
		$\fD(\bCPE)=\fD_\cJ(\bB)\big/\fD(\bB;\cJ)$ is the Lie factor-algebra 
		of {\it derivations of the factor-algebra} $\cA(\bCPE)$; 
	\item 
		$\fD(\bCPE)=\fD_\rV(\bCPE)\oplus_{\cA(\bCPE)}\fD_\rH(\bCPE)$; 
	\item 
		$\fD_\rV(\bCPE)=\{\bar\ze\in\fD(\bCPE)\mid \ze\in\fD_\rV(\bB)\cap\fD_\cJ(\bB)\}$; 
	\item 
		$\fD_\rH(\bCPE)=\{\bar\ze\in\fD(\bCPE)\mid \ze\in\fD_\rH(\bB)\subset\fD_\cJ(\bB)\}$; 
	\item 
		$\Sym(\cA(\bCPE),\fD_\rH(\bCPE))\!
		\!=\!\{\bar\ze\in\fD_\rV(\bCPE)\mid[\bar D_\mu,\bar\ze]\!=\!\bar0\!=\!\cJ, \ \mu\in\rM\}$. 
\end{itemize} 
To write the $\bbZ$-graded $\cA(\bCPE)$-module $\Om_\rH(\bCPE)=\oplus_{q\in\bbZ}\Om^q_\rH(\bCPE)$ 
and its cohomology spaces it is sufficient to replace $\bCE$ with $\bCPE$ 
in all relevant formulas. In particular, 
\begin{equation*} 
	\Om^q_\rH(\bCPE)=\begin{cases} 
	0,                                                & q<0,q>m; \\
	\cA(\bCPE),                                       & q=0; \\
	\Hom_{\cA(\bCPE)}(\w^q\fD_\rH(\bCPE);\cA(\bCPE)), & 1\le q\le m.  
	             \end{cases}
\end{equation*}
The {\it horizontal differential} $d_\rH\in\End_\bbR(\Om_\rH(\bCPE))$, $d_\rH\com d_\rH=0$, 
is defined by the rule, 
\begin{align*} 
	d^q_\rH=d_\rH\big|_{\Om^q_\rH(\bCPE)} : \Om^q_\rH(\bCPE)&\to\Om^{q+1}_\rH(\bCPE), \\ 
	      \bar\om_{\mu_1\dots\mu_q}\cdot d\bar x^{\mu_1}\w\dots\w d\bar x^{\mu_q}
	       &\mapsto 
	       \bar D_{[\mu_0}\bar\om_{\mu_1\dots\mu_q]}\cdot 
	       d\bar x^{\mu_0}\w\dots\w d\bar x^{\mu_q},  
\end{align*}
where $d^{q+1}_\rH\com d^q_\rH=0$, $q\in\bbZ$. 

Now, let us define the global coordinates on the subspace $\bCPE\subset\bB$. 
Namely, we split $\bbI=\bbI_1\cup\bbI'_1$ (see Subsection \ref{NT}) and get   
\begin{itemize}
	\item 
		$\bbR_\bbI=\bbR_{\bbI_1}\times\bbR_{\bbI'_1}$; 
	\item 
		$\bB=\rX\times\bbR^1_{\bbI_0}\times\bbR^1_{\bbI'_0}\times\bbR^{\rN}_\bbI\times
		\bbR_{\bbI_1}\times\bbR_{\bbI'_1}$;
	\item
		$\bCPE=\bCE\cap\bPE=\rX\times\bbR^1_{\bbI_0}\times\bbR^{\rN}_\bbI\times
		\bbR_{\bbI_1}$, i.e., the subspace $\bCPE$ has the global coordinates 
		$
		(x,\bu,\bp)=\{x^\mu,u^1_{\ri_0},u^\al_\ri,p_{\ri_1}\},
		$ 
		where the indices 
		$\mu\in\rM, \ \al\in\rN, \ \ri_0\in\bbI_0, \ \ri_1\in\bbI_1, \ \ri\in\bbI$;  
	\item 
		$\CE=u^\mu_{(\mu)}=0$, \quad 
		$\PE=\De p+u^\la_{(\mu)}u^\mu_{(\la)}\big|_{\bCE}=p_{2(1)}+\Phi(\bu,\bp)=0$;  
	\item
		$\Phi(\bu,\bp)=\De'p+u^\la_{(\mu)}u^\mu_{(\la)}\big|_{\bCE}$, \quad 
		$\De'=\de^{\al\be}D_\al\com D_\be$; 
	\item 
		$u^1_{\ri'_0}=-u^\al_{\ri'_0-(1)+(\al)}$, \ $\ri'_0\in\bbI'_0$, \quad 
		$p_{\ri'_1}=-D_{\ri'_1-2(1)}\Phi(\bu,\bp)$, \ $\ri'_1\in\bbI'_1$;
	\item 
		$D_\ri(u^\la_{(\mu)}u^\mu_{(\la)})
		=\sum_{\rk+\rl=\ri}\binom{\ri}{\rk}u^\la_{\rk+(\mu)}
		u^\mu_{\rl+(\la)}$, \ $\ri=\ri'_1-2(1)\in\bbI$.    
\end{itemize}
In these coordinates, we have: 
\begin{itemize} 
	\item 
		$\cA(\bCPE)=\cC^\infty_{\fin}(\bCPE) \\ 
		\phantom{12345678}
		=\{f\in\cA(\bB)\mid \p_{u^1_{\ri'_0}}f=0, \p_{p_{\ri'_1}}f=0, 
		\ri'_0\in\bbI'_0, \ri'_1\in\bbI'_1\}\subset\cA(\bB)$; 
	\item 
		$\fD_\rV(\bCPE)=\big\{\ze=\ze^1_{\ri_0}\p_{u^1_{\ri_0}}+\ze^\al_\ri\p_{u^\al_\ri}
		+\ze_{\ri_1}\p_{p_{\ri_1}}\big\}$, 
		\quad where  \\ 
		$\ze^1_{\ri_0},\ze^\al_\ri,\ze_{\ri_1}\in\cA(\bCPE)$, 
		$\ri_0\in\bbI_0$, $\ri_1\in\bbI_1$, $\ri\in\bbI$, $\al\in\rN$; 
	\item 
		$D_\mu=\p_{x^\mu}+u^1_{\ri_0+(\mu)}\p_{u^1_{\ri_0}}\!
		+u^\al_{\ri+(\mu)}\p_{u^{\al}_\ri}
		+p_{\ri_1+(\mu)}\p_{p_{\ri_1}}$, \quad  $\mu\in\rM$, \quad where \\ 
		$u^1_{\ri_0+(1)}=-u^\al_{\ri_0+(\al)}$, \ 
		$p_{\ri_0+2(1)}=-D_{\ri_0}\Phi(\bu,\bp)$, \ 
		 \ $\ri_0\in\bbI_0$, \  $\ri_1=\ri_0,\ri_0+(1)$; 
	\item 
		$\Sym(\cA(\bCPE),\fD_\rH(\bCPE)) \\
		\phantom{1234567890}
		=\{\ze=\ev_\rf\in\fD_\rV(\bCPE)\mid [D_\mu,\ev_\rf]=0, 
		\ \mu\in\rM\}$, where  
		\begin{itemize} 
		\item 
		$\ev_\rf=D_{\ri_0}f^1\cdot\p_{u^1_{\ri_0}}+D_\ri f^\al\cdot\p_{u^\al_\ri}
		+D_{\ri_1}f\cdot\p_{p_{\ri_1}}$,
		\item 
		$\rf=(f^\mu,f)\in\cA(\bCPE)^\rM\times\cA(\bCPE)$, 
		\item 
		$D_\mu f^\mu=0, \quad 
			\De f+\ev_\rf\big(u^\la_{(\mu)}u^\mu_{(\la)}\big)=0$, 
		\item 
		$\De=\de^{\la\mu}D_{(\la)+(\mu)}, \quad 
		\ev_\rf\big(u^\la_{(\mu)}u^\mu_{(\la)})
		=2u^\la_{(\mu)}D_\la f^\mu\big|_{u^1_{(1)}=-u^\al_{(\al)}}$;
		\end{itemize} 
	\item 
		$\Om^q_\rH(\bCPE)=\{\om=\om_{\mu_1\dots\mu_q}\cdot 
			dx^{\mu_1}\w\dots\w dx^{\mu_q}\mid 
			\om_{\mu_1\dots\mu_q}\in\cA(\bCPE), \text{+ s-s}\}$, $1\le q\le m$; 
	\item 
		$d^q_\rH=d_\rH\big|_{\Om^q_\rH(\bCPE)} : \Om^q_\rH(\bCPE)\to\Om^{q+1}_\rH(\bCPE)$, 
		where \\ 
		$\om=\om_{\mu_1\dots\mu_q}\cdot dx^{\mu_1}\w\dots\w dx^{\mu_q}\mapsto 
		D_{[\mu_0}\om_{\mu_1\dots\mu_q]}\cdot dx^{\mu_0}\w\dots\w dx^{\mu_q}$. 
\end{itemize}  
Like it was done in the Subsection \ref{dfs}, to calculate the cohomology spaces 
$H^q_\rH(\bCPE)$ we introduce the auxiliary complex $\{\Te^q,d^q_\Te \mid q\in\bbZ\}$ 
(see page \pageref{AC1}). The only changes are \label{AC2}
\begin{itemize} 
	\item 
		the set $\bCE$ is replaced with the set $\bCPE$;   
	\item 
		the basic derivations $D_1,D_\al$ are as in the present section; 
	\item 
		the operator $\tilde D_1=D_1+\rf^*$, 
	\item 
		$D_1=\p_{x^1}+\ev'_\rf$,\quad where \quad $\rf=(f^1_0,f^\al_{i^1},f_0,f_1)$, \\ 
		$\ri_0=(0,\rj)$, \ $\ri_1=(0,\rj),(1,\rj)$, \ $\ri=(i^1,\rj)$, \ $i^1\in\bbZ_+$,  
		\ $\rj\in\bbJ$; 
	\item 
		$\ev'_\rf=D_\rj f^1_0\cdot\p_{u^1_{0,\rj}}+D_\rj f^\al_{i^1}\cdot\p_{u^\al_{i^1,\rj}}
		+D_\rj f_0\cdot\p_{p_{0,\rj}}+D_\rj f_1\cdot\p_{p_{1,\rj}}$; 
	\item 
		$f^1_0=-u^\al_{(\al)}$, \ $f^\al_{i^1}=u^\al_{i^1+1,0}$, \ 
		$f_0=p_{1,0}$, \ $f_1=-\Phi(\bu,\bp)$; 
	\item 
		$\De'=\de^{\al\be}D_{(\al)+(\be)}$, \quad 
		$u^\la_{(\mu)}u^\mu_{(\la)}=u^\al_{(\al)}u^\be_{(\be)}+2u^1_{(\al)}u^\al_{(1)}
		+u^\al_{(\be)}u^{\be}_{(\al)}$; 
	\item 
		$\chi\mapsto\rf^*\chi$, \quad 
		$\chi=(\chi^0_1,\chi^{i^1}_\al,\chi^0,\chi^1)$, 
		where the set $\chi$ is {\it finite}, i.e., only 
		finitely many of its elements are nonzero,
	\item 
		$(\rf^*\chi)=((\rf^*\chi)^0_1,(\rf^*\chi)^{i^1}_\al,
		(\rf^*\chi)^0,(\rf^*\chi)^1)$; 
	\item 
		$(\rf^*\chi)^0_1=2D_\al(u^\al_{(1)}\chi^1)$, \\ 
		$(\rf^*\chi)^{i^1}_\al=\de^{i^1}_0D_\al\chi^0_1+\chi^{i^1-1}_\al \\
		\phantom{123456}+2\big(\de^{i^1}_0D_\al(u^\be_{(\be)}\chi^1)
		-\de^{i^1}_1u^1_{(\al)}\chi^1+\de^{i^1}_0D_\be(u^{\be}_{(\al)}\chi^1)\big)$, \\ 
		$(\rf^*\chi)^0=-\De'\chi^1$, \quad $(\rf^*\chi)^1=\chi^0$, 
\end{itemize}
where $\al,\be\in\rN$, \ $\la,\mu\in\rM$, \ $i^1\in\bbZ_+$, \ $\rj\in\bbJ$. 

The resulting system $(D_1+\rf^*)\chi=0$ for unknown function $\chi$ with a finite number 
of nonzero components reduces to 
\begin{align*}  
	&D_1\chi^0_1+(\rf^*\chi)^0_1=0, \quad 
		D_1\chi^{i^1}_\al+(\rf^*\chi)^{i^1}_\al=0, \\
	&D_1\chi^0+(\rf^*\chi)^0=0, \quad D_1\chi^1+(\rf^*\chi)^1=0.  
\end{align*}
The equation 
\begin{align*}
	D_1\chi^{i^1}_\al+(\rf^*\chi)^{i^1}_\al
		&=D_1\chi^{i^1}_\al+\de^{i^1}_0D_\al\chi^0_1+\chi^{i^1-1}_\al \\
		&+2\big(\de^{i^1}_0D_\al(u^\be_{(\be)}\chi^1)
			-\de^{i^1}_1u^1_{(\al)}\chi^1+\de^{i^1}_0D_\be(u^{\be}_{(\al)}\chi^1)\big)=0  
\end{align*}
for $i^1\ge2$, gives $\chi^{i^1}_\al=0$ for $i^1\ge1$,   
for $i^1=1$ it gives $\chi^0_\al=2u^1_{(\al)}\chi^1$, 
and for $i^1=0$ it takes the form  
$D_1\chi^0_\al+D_\al\chi^0_1+2\big(D_\al(u^\be_{(\be)}\chi^1)
+D_\be(u^\be_{(\al)}\chi^1)\big)=0$. 
Further, the equation $D_1\chi^1+(\rf^*\chi)^1=D_1\chi^1+\chi^0=0$ 
gives $\chi^0=-D_1\chi^1$, 
the equation $D_1\chi^0+(\rf^*\chi)^0=-D_{2(1)}\chi^1-\De'\chi^1=-\De\chi^1=0$ 
gives $\De\chi^1=0$, 
and the last one equation $D_1\chi^0_1+(\rf^*\chi)^0_1=D_1\chi^0_1
+2D_\al(u^\al_{(1)}\chi^1)=0$ gives $D_1\chi^0_1+2D_\al(u^\al_{(1)}\chi^1)=0$. 

After some algebra we get 
\begin{lemma}\label{L2} 
The system $(D_1+\rf^*)\chi=0$, $\chi=(\chi^0_1,\chi^{i^1}_\al,\chi^0,\chi^1)$, 
reduces to 
\begin{equation} 
	\chi^0_\al=2u^1_{(\al)}\chi^1, \quad \chi^{i^1}_\al=0, \quad 
	\chi^0=-D_1\chi^1, \quad  i^1\ge1, \ \al\in\rN, 
\end{equation} 
where $\chi^1,\chi^0_1$ are solutions of the system 
\begin{align} 
	&\De\chi^1=0,  \quad 
	\big(u^\mu_{(1)}D_\mu D_\al-u^\mu_{(\al)}D_\mu D_1\big)\chi^1=0, 
	\quad \al\in\rN,\label{00} \\
	&D_1\chi^0_1=-2D_\al(u^\al_{(1)}\chi^1), \   
	 D_\al\chi^0_1=-2\big(u^\mu_{(\al)}D_\mu\chi^1+D_\al(u^\be_{(\be)}\chi^1)\big),
	 	\ \al\in\rN \label{01}. 
\end{align} 
\end{lemma} 

\begin{rem} 
The right subsystem of the system (\ref{00}) is the 
{\it non-trivial compatibility condition} of the system (\ref{01}).
\end{rem}

\begin{rem}\label{R3} 
The system (\ref{00}) has the trivial solution $\chi^1=0$, 
then the system (\ref{01}) gives $\chi^0_1=\const$, and we get the solution 
$\chi=(1,0,0,0)$ of the full system $(D_1+\rf^*)\chi=0$ 
(see the similar result in Subsection \ref{dfs}).
\end{rem} 

\begin{theorem} 
The linear spaces of the cohomologies of the differential algebra 
$(\cA(\bCPE),\fD_\rH(\bCPE))$ are 
\text{\rm(cf., Theorem \ref{T1} and Theorem \ref{TCE})}  
\begin{equation*} 
	H^q_\rH(\bCPE)=\begin{cases} 
		0,                              & q<0,1\le q\le m-2, q>m; \\
		\bbR,                           & q=0; \\ 
		\ke D^{m-1}_1\simeq\cS\cap\cH,  & q=m-1; \\
		H^{m-1}(\Te)\big/\im D^{m-1}_1, & q=m; 
		   	   \end{cases}
\end{equation*} 
where 
\begin{itemize}
	\item 
		$\cS=\Sol(D_1+\rf^*)$ is the linear space of solutions 
		$\chi=(\chi^0_1,\chi^{i^1}_\al,\chi^0,\chi^1)$ of the linear system 
		$D_1\chi+\rf^*\chi=0$ \text{\rm (the mapping $\rf^*$ is defined above)}; 
	\item 
		$\cH=\{\chi=(\chi^0_1,\chi^{i^1}_\al,\chi^0,\chi^1)\mid \chi_*=\chi^*\}$ is 
		the Helmholtz space of the differential algebra $(\cA(\bCPE),\fD_\rH(\bCPE))$; 
	\item 
		$H^{m-1}_\rH(\bCPE)=\cS\cap\cH$, \quad 
		$[J^\mu\cdot d_\mu x]\mapsto
		(\de_{u^1_0}J^1,\de_{u^\al_{i^1}}J^1,\de_{p_0}J^1,\de_{p_1}J^1)$ 
		\text{\rm (remind, cohomologies are defined up to isomorphisms)}. 
\end{itemize}	
\end{theorem} 

\begin{rem}\label{R5} 
In particular, $H^{m-1}_\rH(\bCPE)\ni[u^\mu d_\mu x]$, see Remark \ref{R3}. 
Moreover, the additional constraint $\PE=\De p+u^\la_{(\mu)}u^\mu_{(\la)}=0$ 
generates the additional cohomology $[F^\mu d_\mu x]\in H^{m-1}_\rH(\bCPE)$, 
where $F^\mu=-u^\la u^\mu_{(\la)}+\nu\De u^\mu+p^{(\mu)}$, since here 
\begin{align*} 
	D_\mu F^\mu
		&=-\big(\De p+u^\la_{(\mu)}u^\mu_{(\la)}\big)
			+\big(-u^\la D_\la u^\mu_{(\mu)}+\nu\De u^\mu_{(\mu)}\big) \\
		&=-\PE+\big(\nu\De-u^\la D_\la\big)\CE=0
\end{align*}
due to the constraints $\PE=0$ and $\CE=0$. 
Here, $p^{(\mu)}=D^\mu p=\de^{\mu\la}D_\la p=D_\mu p=p_{(\mu)}$ 
(the euclidean metrics). 
\end{rem} 

\begin{rem} 
Every symmetry $\ev_\rf\in\Sym(\cA(\bCPE),\fD_\rH(\bCPE))$, $\rf=(f^\mu,f)$, 
generates (possibly trivial) cohomology $[f^\mu d_\mu x]\in H^{m-1}_\rH(\bCPE)$, 
because in this case $D_\mu f^\mu=0$ by definition. 
\end{rem}

\section{Navier-Stokes equations as evolution in the space $\bCPE$.} 

\subsection{Evolution in the space $\bCPE$.} 

The evolution in the space 
$$ 
	\bCPE=
	\rT\times\rX\times\big(\bbR^1_{\bbI_0}\times\bbR^{\rN}_\bbI\big)\times\bbR_{\bbI_1}
	=\big\{t,\rx=(x^\mu),\bu=(u^1_{\ri_0},u^\al_\ri),\bp=(p_{\ri_1})\big\}
$$ 
($\rM=2,3$) is governed by an evolution derivation 
$
	D_t=\p_t+\ev_\rE, 
$
where 
\begin{itemize} 
	\item 
		we added the time variable $t\in\rT=\bbR$, one may assume that it was present 
		from the start as a parameter; 
	\item 
		$\rE=(E^\mu,E)\in\cA^\rM(\bCPE)\times\cA(\bCPE)$; 
	\item 
		$D_\mu E^\mu=0, \quad  
		\De E+2\big(u^\la_{(\mu)}D_\la E^\mu\big)\big|_{u^1_{(1)}
		=-u^\al_{(\al)}}=0, \quad \De=\de^{\la\mu}D_{(\la)+(\mu)}$; 
	\item 
		$\ev_\rE
		=D_{\ri_0}E^1\cdot\p_{u^1_{\ri_0}}+D_\ri E^\al\cdot\p_{u^\al_\ri}
		+D_{\ri_1}E\cdot\p_{p_{\ri_1}}\in\Sym(\cA(\bCPE),\fD_\rH(\bCPE))$
\end{itemize}
(see Subsection \ref{ADC}, remind also Proposition \ref{sce} and Remark \ref{rce}). 

\begin{rem} 
In particular, $D_t u^\mu=F^\mu$, $\mu\in\rM$, and $D_t p=F$, while 
$$ 
	D_t\CE\!=\!D_t\big(D_\mu u^\mu\big)\!=\!D_\mu E^\mu, \quad 
	D_t\PE\!=\!D_t\big(\De p+u^\la_{(\mu)}u^\mu_{(\la)}\big)
	\!=\!\De E+2u^\la_{(\mu)}D_\la E^\mu.
$$
\end{rem} 

There is defined the differential algebra $(\cA(\bCPE),\fD_\rE(\bCPE))$, 
where 
\begin{itemize} 
	\item 
		$\fD(\bCPE)=\fD_\rV(\bCPE)\oplus_{\cA(\bCPE)}\fD_\rE(\bCPE)$; 
	\item 
		$\fD_\rV(\bCPE)=\{\ze=\ze^1_{\ri_0}\p_{u^1_{\ri_0}}+\ze^\al_\ri\p_{u^\al_\ri}
		+\ze_{\ri_1}\p_{p_{\ri_1}}\mid
		\ze^1_{\ri_0},\ze^\al_\ri,\ze_{\ri_1}\in\cA(\bCPE)\}$; 
	\item 
		$\fD_\rE(\bCPE)$ has the $\cA(\bCPE)$-basis 
		$\{D_t,D_\mu\mid \mu\in\rM\}$, 
		the basic time derivation  $D_t=\p_t+\ev_\rE$, 
		$[D_t,D_\mu]=0$, $\mu\in\rM$, so  
		$$ 
		\fD_\rE(\bCPE)=\big\{\ze=\ze^t D_t+\ze^\mu D_\mu \ \big| 
		\ \ze^t,\ze^\mu\in\cA(\bCPE)\big\}.
		$$ 
\end{itemize}  
The Lie algebra of symmetries here is 
\begin{align*} 
	\Sym(\cA(\bCPE),&\fD_\rE(\bCPE)) \\ 
    &=\big\{\ze=\ev_\rf\in\Sym(\cA(\bCPE,\fD_\rH(\bCPE))
    \mid [D_t,\ev_\rf]=0\big\},	
\end{align*} 
where the condition $[D_t,\ev_\rf]=0$ reduces to the equation 
$(D_t-\rE_*)\rf=0$. 

In more detail, here 
\begin{itemize} 
	\item
		$\rE_* : \cA(\bCPE)^\rM\times\cA(\bCPE)\to\cA(\bCPE)^\rM\times\cA(\bCPE)$, 
	\item
		$\rf=(f^\mu,f)\mapsto\rF_*\rf=((\rF_*\rf)^\mu,(\rF_*\rf))$,  
	\item 
		$(\rE_*\rf)^\mu
			=\p_{u^1_{\ri_0}}E^\mu\cdot D_{\ri_0}f^1
			+\p_{u^\al_\ri}E^\mu\cdot D_\ri f^\al
			+\p_{p_{\ri_1}}E^\mu\cdot D_{\ri_1}f$, 
	\item  
		$(\rE_*\rf)
			=\p_{u^1_{\ri_0}}E\cdot D_{\ri_0}f^1
			+\p_{u^\al_\ri}E\cdot D_\ri f^\al 
			+\p_{p_{\ri_1}}E\cdot D_{\ri_1}f$.
\end{itemize}

As above, to study the cohomology spaces 
$H^q_\rE(\bCPE)=\ke d^q_\rE\big/\im d^{q-1}_\rE$, $q\in\bbZ$, 
of the differential algebra $(\cA(\bCPE),\fD_\rE(\bCPE))$ 
we shall use the decomposition (see pages \pageref{AC1} and \pageref{AC2}):   
\begin{itemize} 
	\item 
		$\Om^q_\rH=\Om^q_\rH(\bCPE) \\ 
		\phantom{123}
		=\big\{\om^q_\rH=\om_{\mu_1\dots\mu_q}\cdot dx^{\mu_1}\w\dots
		\w dx^{\mu_q} \ \big| \ \om_{\mu_1\dots\mu_q}
		\in\cA(\bCPE), \text{+ s-s}\big\}$;
	\item 
		$\Om^q_\rE=\Om^q_\rE(\bCPE)=dt\w\Om^{q-1}_\rH
		\oplus_{\cA(\bCPE)}\Om^q_\rH$, \quad $q\in\bbZ$; 
	\item 
		$0\to\Om^{q-1}_\rH\xrightarrow{\io^{q-1}}\Om^q_\rE\xrightarrow{\pi^q}\Om^q_\rH\to0$, \\ 
		\phantom{$0\to$}
		$\om^{q-1}_\rH\mapsto\om^{q}_\rE=\io^{q-1}\om^{q-1}_\rH=(-1)^{q-1}dt\w\om^{q-1}_\rH$, \\ 
		\phantom{$0\to\Om^{q-1}\xrightarrow{\io^{q-1}}$}
		$\om^q_\rE=dt\w\om^{q-1}_\rH+\om^q_\rH\mapsto\pi^q\om^q_\rE=\om^q_\rH$; 
	\item 
		$d^q_\rE=d^q_t+d^q_\rH : \Om^q_\rE\to\Om^{q+1}_\rE$, 
		\quad $d_t=dt\w D_t$, \quad $d_\rH=dx^\mu\w D_\mu$, \\  
		$\om^q_\rE=dt\w\om^{q-1}_\rH+\om^q_\rH\mapsto d_\rE\om^q_\rE
			=dt\w(D_t\om^q_\rH-d_\rH\om^{q-1}_\rH)+d_\rH\om^q_\rH$.  
\end{itemize} 
These constructions lead to the commutative diagram with the long exact sequence of the cohomology spaces 
\begin{align*} 
	                                         0\to&H^0_\rE(\bCPE)\to 
	                                  H^0_\rH(\bCPE)
	                    \xrightarrow{D^0_t} \\ 
	         \xrightarrow{D^0_t}H^0_\rH(\bCPE)\to&H^1_\rE(\bCPE)
	\to H^1_\rH(\bCPE)\xrightarrow{D^1_t}\dots \\
\dots\xrightarrow{D^{m-2}_t}H^{m-2}_\rH(\bCPE)\to&H^{m-1}_\rE(\bCPE)
 	        \to H^{m-1}_\rH(\bCPE)\xrightarrow{D^{m-1}_t} \\          
 	 \xrightarrow{D^{m-1}_t}H^{m-1}_\rH(\bCPE)\to&H^m_\rE(\bCPE)
 	 \to H^m_\rH(\bCPE)\xrightarrow{D^m_t} \\ 
	         \xrightarrow{D^m_t}H^m_\rH(\bCPE)\to&H^{m+1}_\rE(\bCPE)\to0,    
\end{align*}
where $D^q_t : H^q_\rH(\bCPE)\to H^q_\rH(\bCPE)$ 
by the component-wise rule  
$$ 
	D^q_t[\om_{\mu_1\dots\mu_q}\cdot dx^{\mu_1}\w\dots\w dx^{\mu_q}]
		=[(D_t\om_{\mu_1\dots\mu_q})\cdot dx^{\mu_1}\w\dots\w dx^{\mu_q}].
$$

\begin{theorem} 
The linear spaces of cohomologies of the differential algebra $(\cA(\bCPE);\fD_\rE(\bCPE))$ are 
\begin{equation*}
 H^q_\rE(\bCPE)=\begin{cases} 
 	0,                           & q<0,1\le q\le m-2,q> m+1; \\
 	\bbR,                        & q=0; \\
 	\ke D^{m-1}_t,               & q=m-1; \\ 
 	H^m_\rH(\bCPE)\big/\im D^m_t, & q=m+1; 
               \end{cases}
\end{equation*} 
while in the case $q=m$ one has $H^m_\rE(\bCPE)\big/\im H^{m-1}_\rH(\bCPE)=\ke D^m_t$.
\end{theorem} 
\begin{proof} 
Indeed, the exact subsequence $0\to H^0_\rE(\bCPE)\to\ke D^0_t\to 0$ gives $H^0_\rE(\bCPE)=\bbR$ 
(remind, we added the time variable $t\in\rT=\bbR$, so now  
$H^0_\rH(\bCPE)=\cT=\{\phi(t)\in\cC^\infty(\bbR)\}$). 
Further, the exact subsequence 
$$
	H^0_\rH(\bCPE)\xrightarrow{D^0_t}H^0_\rH(\bCPE)\to H^1_\rE(\bCPE)\to H^1_\rH(\bCPE)=0
$$ 
gives $H^1_\rE(\bCPE)=\cT\big/\im D^0_t=0$. 
Now, the exact subsequence 
$$
	H^{q-1}_\rH(\bCPE)\to H^q_\rE(\bCPE)\to H^q_\rH(\bCPE)
$$ 
gives $H^q_\rE(\bCPE)=0$ for $2\le q\le m-2$ because in this case $H^{q-1}_\rH(\bCPE)=H^q_\rH(\bCPE)=0$. 
Then, the exact subsequence 
$$ 
	0=H^{m-2}_\rH(\bCPE)\to H^{m-1}_\rE(\bCPE)\to H^{m-1}_\rH(\bCPE)
		\xrightarrow{D^{m-1}_t}\dots  
$$ 
gives $H^{m-1}_\rE(\bCPE)=\ke D^{m-1}_t$. 
At last, for $q=m,m+1$ the statements follow from the exact subsequence 
$H^{m-1}_\rH(\bCPE)\to H^m_\rE(\bCPE)\to\ke D^m_t\to0$ and from the exact subsequence 
$0\to\im D^m_t\to H^m_\rH(\bCPE)\to H^{m+1}_\rE(\bCPE)\to0$. 
\end{proof} 

\begin{rem} 
To calculate $\ke D^{m-1}_t$ one may use the technics from lemmas \ref{L1} and \ref{L2}.
The cases $q=m,m+1$ in this approach are uninformative. They demand a special study.  
\end{rem}

\subsection{Navier-Stokes equations as the evolution with constraints.} 

We consider the Navier-Stokes system (\ref{NS1})-(\ref{NS4}) as the evolution 
process governed by the equation (\ref{NS1}) in the divergence-free space with the inner constraint (\ref{NS4}).   

The algebraic counterpart of the equation (\ref{NS1}) is the symmetry 
$$ 
	\ev_\rE=D_{\ri_0}E^1\cdot\!\p_{u^1_{\ri_0}}+D_\ri E^{\al}\cdot\p_{u^{\al}_\ri}
	+D_{\ri_1}E\cdot\p_{p_{\ri_1}}\in\Sym(\cA(\bCPE),\fD_\rH(\bCPE)), 
$$ 
where 
\begin{itemize} 
	\item 
		$\rE=(E^\mu,E)\in\cA(\bCPE)^\rM\times\cA(\bCPE)$; 
	\item 
		$E^\mu=-u^\la\na_\la u^\mu+\nu\De u^\mu-\na^\mu p
		=-u^\la u^\mu_{(\la)}+\nu\De u^\mu-p_{(\mu)}$;  
	\item 
		$u^1_{(1)}=-u^{\al}_{(\al)}$, \quad
		$\De u^\mu=\sum_\la u^\mu_{2(\la)}$, \quad  
		$\De u^1=-u^{\al}_{(1)+(\al)}+\sum_{\al}u^1_{2(\al)}$; 
	\item 
		$\na^\mu=\de^{\mu\nu}\na_\nu=\na_\mu$ (the euclidean metrics), 
		so $p^{(\mu)}=p_{(\mu)}$;
	\item 
		$E$ to be defined from the condition $\ev_\rE\in\Sym(\cA(\bCPE),\fD_\rH(\bCPE))$.  
\end{itemize} 
One can easily check that here $D_\mu E^\mu=0$. 
On the other hand the condition 
\begin{equation}
	\De E+\ev_\rE(u^\la_{(\mu)}u^\mu_{(\la)})=\De E+2u^\la_{(\mu)}D_\la E^\mu=0
\end{equation}
is the {\it Poisson equation for the component $E$} (cf., Remark \ref{IC}).  

Thus, for the algebraic analysis of the Navier-Stokes equations, one can use the technique outlined in the previous subsection. 

\subsection{Conclusion.}  

It can be seen from the above constructions that the Navier-Stokes equations 
are subject to meaningful analysis within the framework of the algebraic approach 
to differential equations. The resulting equations for finding algebraic characteristics 
of Navier-Stokes equations, such as symmetries and cohomologies, are essentially complicated. 
One may hope to find their partial solutions at least, especially using analytical 
computational packets (Mathematica, for example).

\end{document}